\documentclass[prl,superscriptaddress,twocolumn]{revtex4}
\usepackage[utf8]{inputenc}
\usepackage{amsmath}
\usepackage{amsfonts}
\usepackage{amssymb}
\usepackage{graphicx}
\usepackage{braket}
\usepackage{dsfont}
\usepackage{mathtools}
\usepackage{framed}
\usepackage{float}
\usepackage{physics}
\usepackage{epstopdf}
\usepackage{physics}
\usepackage{multirow}
\usepackage{amsthm}
\usepackage{color}
\usepackage{tabularx}
\usepackage{algorithm}
\floatname{algorithm}{Protocol}
\usepackage{algorithmic}
\usepackage{amsthm}
\usepackage{array}
\usepackage{subfigure}
\usepackage{makecell}

\usepackage{hyperref}
\hypersetup{
    colorlinks=true, 
    linkcolor=red, 
    linktoc=all 
}

\begin{document}
\title{Authenticated teleportation and verification in a noisy network}
\author{Anupama Unnikrishnan}
\affiliation{Clarendon Laboratory, University of Oxford, Oxford OX1 3PU, UK}
\author{Damian Markham}
\affiliation{LIP6, CNRS, Sorbonne Universit\'e, 4 Place Jussieu, 75005 Paris}

\newtheorem{theorem}{Theorem}
\newtheorem{lemma}{Lemma}
\newtheorem{corollary}{Corollary}
\newtheorem{definition}{Definition}
\newenvironment{proof-sketch}{%
\renewcommand{\proofname}{Proof sketch}\proof}{\endproof}
\newcommand\Myperm[2][^n]{\prescript{#1\mkern-2.5mu}{}P_{#2}}
\newcommand\Mycomb[2][^n]{\prescript{#1\mkern-0.5mu}{}C_{#2}}

\begin{abstract}
Authenticated teleportation aims to certify the transmission of a quantum state through teleportation, even in the presence of an adversary. This scenario can be pictured in terms of an untrusted source distributing a Bell state between two parties who wish to verify it using some simple tests. We propose a protocol that achieves this goal in a practical way, and analyse its performance and security when the parties have noisy measurement devices. Further, we model a realistic experimental scenario where the state is subject to noise. We finally apply our analysis to the verification of graph states with noisy measurement devices.
\end{abstract}


\maketitle

\section{Introduction}

By the quantum teleportation protocol, a sender can transmit a quantum state to a receiver using only a shared Bell state and classical communication \cite{Bennett1993}. Over the years, teleportation has grown to be a fundamental protocol in the field of quantum information, and further it has been integrated into numerous other protocols as an important step \cite{Hillery1999, Christandl2005, Crepeau2002}. In this work, we consider \textit{authenticated teleportation}, where one can certify the teleportation procedure for sending quantum messages without trusting the shared entanglement \cite{Unnikrishnan2019a}.

The cryptographic task of authentication is concerned with two parties that share a private classical key, who wish to exchange a message with the guarantee that it has not been tampered with by an adversary that controls their communication channel.
Formulating this problem in the quantum setting, the parties Alice and Bob may share an insecure or imperfect quantum channel, which they would like to use to send messages in the form of quantum states. If they had in their possession a perfect shared Bell state, Alice could send a quantum state to Bob using the quantum teleportation procedure. Since the Bell state forms a quantum channel between only Alice and Bob, it follows that only Alice could have sent the message that Bob receives, and an adversary cannot tamper with it during the transmission. However, in authentication, the only prerequisite is that our parties share a private classical key. A quantum authentication scheme must then provide a way to distribute and test Bell states before using them as a quantum channel through which quantum messages can be sent.

Barnum et al. defined a framework for quantum authentication and additionally proposed the first scheme for authentication of quantum messages, based on error-correcting codes, in \cite{Barnum2002}. Further work on authentication of quantum messages has built upon this, focusing on aspects such as proving the composability of the above protocol \cite{Hayden2016}, allowing the key to be recycled if the message is accepted \cite{Hayden2016, Portmann2017}, and even in the case where there is some tampering by the adversary \cite{Portmann2017}. While the exponential security scaling of such protocols is highly desirable, it comes at the cost of requiring levels of entanglement that increase with the security parameter. In practice, this becomes infeasible.

Our approach, on the other hand, is a simple yet practical way of achieving authenticated quantum communication. We consider an untrusted, non-iid source who claims to be creating Bell states that the parties wish to use for teleportation. The parties then merely need to test many separate copies of the Bell state before using one for the teleportation procedure. In this way, the quantum channel is authenticated before transmitting a message. While such a straightforward approach adversely affects the level of security one can expect, our focus is on methods that are not experimentally demanding; in fact, our scheme can be easily implemented with current linear optical setups. 

We use and extend a simple verification protocol from \cite{Marin}, based on measuring stabilisers of the state, which can be applied to more complicated situations such as verifying graph states \cite{Markham2018}. Rather than requiring a large entangled state over multiple copies as in the quantum authentication scheme of \cite{Barnum2002}, our protocol simply requires many separate copies of the state.

As we are interested in how this may be applied to practical quantum networks, we consider the influence of noise on our authentication procedure. We demonstrate the tradeoff between how well our protocol works and its likelihood of failure, in such a noisy setting. 

Finally, our protocol and analysis, being formulated in terms of measuring stabilisers, is then easily extended to the verification of graph states in a realistic, noisy scenario. Graph states are a family of states which act as useful resources for quantum networks and computation, with applications including error correction \cite{Bella}, fault tolerant quantum computation \cite{Raussendorf2006, Nielsen2005}, cryptography \cite{Markham, Bell} and sensing \cite{Shettell2019, Friis2017}. As such, their verification is important and has been studied in various works \cite{Markham2018, Takeuchi2019, Hayashi2015, Takeuchi2018, Hayashi2019, Unnikrishnan2020}. This work analyses more deeply the effect of noise in protocols such as that presented in \cite{Markham2018}.

\section{Quantum authentication}
We start by outlining the authentication framework of Barnum et al. \cite{Barnum2002}, which we use to discuss and define security. A quantum authentication scheme is comprised of a shared classical key, $\kappa$, known only to Alice and Bob and chosen uniformly from the set of keys $\mathcal{K}$, and corresponding operators $A_\kappa, B_\kappa$. Alice sends a message $\ket{\psi}$ to Bob by encoding with her operator $A_\kappa$. The output of Bob, after decoding with his operator $B_\kappa$, is the resulting message, in addition to a classical register with basis states $\ket{ACC},\ket{REJ}$ which indicate acceptance or rejection, respectively. 
Based on the definitions in \cite{Barnum2002}, we introduce the following.

\begin{definition}
We define the following security properties of a protocol for quantum authentication:
\begin{itemize}

\item \textnormal{Completeness}: 
The protocol has \emph{completeness $c$} if, when there is no adversarial intervention, the state accepted by Bob will be the same as that sent by Alice up to $c$; that is, for all $\kappa \in \mathcal{K}$, 

\begin{align}
\text{Tr} \Big[ \big( \ket{\psi} \bra{\psi} \otimes \ket{ACC} \bra{ACC} \big) \big( B_\kappa \big[ A_\kappa (\ket{\psi} \bra{\psi}) \big] \big) \Big] 
\geq c.
\end{align}

\item \textnormal{Soundness}: 
If the adversary's intervention is characterised by $\mathcal{O}$, the resulting output state on Bob's side after the protocol is 
\begin{align}
\rho_{out} =  \frac{1}{\abs{\mathcal{K}}} \sum_\kappa B_\kappa \Big[ \mathcal{O} \big[ A_\kappa (\ket{\psi} \bra{\psi}) \big] \Big],
\end{align}
 and the projector associated with failure is given by 
 \begin{align}
 P_{fail} = (\mathds{1} - \ket{\psi} \bra{\psi}) \otimes \ket{ACC} \bra{ACC},
 \end{align}
then, the protocol has \emph{soundness $\epsilon$} if 
\begin{align}
\text{Tr} (P_{fail} \rho_{out}) \leq \epsilon.
\end{align}
\end{itemize}
\label{def:qas}
\end{definition}

In the notation of \cite{Barnum2002}, a protocol is $\epsilon$-secure if it has completeness $1$ and soundness $\epsilon$.
The completeness condition is associated with an honest run of the protocol; in such an ideal case, the protocol should work perfectly.
The soundness condition represents the failure of the protocol in the presence of an adversary, where by failure we mean that the accepted state lies in the orthogonal subspace to the ideal. It then tells us that, despite adversarial intervention, the maximum probability of failure of the protocol is $\epsilon$.
The smaller the $\epsilon$, the more secure the protocol. 
For example, in the authentication procedure of \cite{Barnum2002}, $\epsilon$ decreases exponentially with an increase in the size of the encoding, scaling as $2^{-S}$.

\section{Protocol}
Before describing our protocol, let us state the communication scenario we consider. The two honest parties, Alice and Bob, can perform local operations and measurements. A source, who may be dishonest, is asked to produce Bell states $\ket{\Phi^+}= \frac{1}{\sqrt{2}} (\ket{00} + \ket{11})$ that our parties wish to use for quantum teleportation. The source is free to produce different states in each round. 
The measurement devices of Alice and Bob are trusted, but may be noisy. In addition, we assume an authenticated classical channel between the parties. This can be substituted by shared random secret keys, ensuring that our model lies in the authenticated communication setting as in \cite{Barnum2002}.

Our protocol for authenticated teleportation is based on the work of Marin and Markham in \cite{Marin}, where they tackle quantum secret sharing over untrusted channels. We will start by adapting their protocol to the scenario of two parties wishing to authenticate their quantum channel for teleportation, and then address noise in the network, finally extending our noise analysis to the certification of graph states.

We will formulate our protocol in terms of stabilisers, so that our analysis can be extended to other stabiliser states in a straightforward way. 
For the Bell state $\ket{\Phi^+}$, the stabiliser group is given by $\{\mathds{1} \otimes \mathds{1}, \sigma_X \otimes \sigma_X, - \sigma_Y \otimes \sigma_Y, \sigma_Z \otimes \sigma_Z\}$. Measuring any of these stabilisers on $\ket{\Phi^+}$ will always give a $+1$ outcome, and further, it is the only state for which this holds.

Based on this, we define an authenticated teleportation scheme given in Protocol \ref{alg:authteleptrust}. This protocol protects the parties against an untrusted source who creates a state that is not $\ket{\Phi^+}$ and attempts to trick the parties into using it for teleportation. The idea behind the protocol is simply that by asking the source for multiple copies of the state, and randomly choosing whether to test or use it, the source is forced to behave honestly in order to avoid being caught. 

This corresponds to a quantum authentication scheme as in Definition \ref{def:qas}, where the shared random classical key $\kappa$ comprises of random strings that specify $r$ and the choice of stabiliser to measure in each round $i \neq r$. Then, depending on this key $\kappa$, the parties apply their corresponding operations. 
We will demonstrate that our protocol fulfils the conditions of completeness and soundness. 

As mentioned previously, we are interested in analysing the performance of the protocol in a realistic, noisy network. 
We will incorporate noise into the authentication part of our scheme, but once the authenticated quantum channel has been established, we will assume that the parties can perform the teleportation perfectly.
We first consider how noise in Alice and Bob's measurement devices affects the security of the protocol, and further investigate the operation of the protocol when the source provides noisy states. Along the way, we will adapt the protocol to suit our purposes. 

\begin{algorithm}[t]
\caption{\textsc{Authenticated teleportation}}
\begin{flushleft}
\textit{Input}: Security parameter $S$. \\
\textit{Goal}: Alice teleports state $\ket{\psi}$ to Bob through an authenticated channel.
\end{flushleft}
\begin{algorithmic}[1]
\STATE An untrusted source generates $S$ copies of the Bell state, and sends the shares of each to Alice and Bob.  \\ \
\STATE Alice chooses a random $r \in \{1, 2, ..., S\}$ and sends $r$ to Bob.
 \\ \
\STATE For all copies $i \neq r$, Alice randomly chooses to measure either $\sigma_X, \sigma_Y, \sigma_Z$ on her part of the state. She tells Bob which operator she measured, and her measurement outcome. \\ \
\STATE For all copies $i \neq r$, Bob measures the same operator as Alice. For each copy, if the product of their measurement outcomes is $+1$ (or $-1$ when measuring $\sigma_Y$), they pass the test, otherwise they fail. \\ \
\STATE If \textit{all} tests on copies $i \neq r$ were passed, the parties ACCEPT. Otherwise, they REJECT. \\ \
\STATE The parties use copy $r$ as the entangled state to teleport $\ket{\psi}$ to Bob. 
\end{algorithmic}
\label{alg:authteleptrust}
\end{algorithm}

\section{Security analysis}
We first derive security bounds for Protocol \ref{alg:authteleptrust}, where the source can supply any state in the hope of cheating. We start by assuming Alice and Bob can do perfect measurements, and then extend this analysis to the case where their measurement devices are noisy. 

\subsection{Perfect measurements}
For the case of perfect measurements, we first go through the security proof from \cite{Marin}, since we will be expanding on this in later sections.
Let $\Pi$ be the projector onto $\ket{\Phi^+}$, given by $\Pi = \ket{\Phi^+}\bra{\Phi^+}$. This can be written in terms of the stabilisers as 
\begin{equation}
\Pi = \frac{1}{4} [\mathds{1} \otimes \mathds{1} + \sigma_X \otimes \sigma_X + \sigma_Z \otimes \sigma_Z + \sigma_Z \sigma_X \otimes \sigma_Z \sigma_X],
\end{equation}
where from now onwards we will write $-\sigma_Y \otimes \sigma_Y$ as $\sigma_Z \sigma_X \otimes \sigma_Z \sigma_X$. 
 The projector onto the +1 eigenstate of the stabiliser $\sigma_X \otimes \sigma_X$ (passing the test) is given by
$
\frac{\mathds{1} \otimes \mathds{1} + \sigma_X \otimes \sigma_X}{2}
$, and similarly for the others.
Thus, the POVM element for passing a test is
\begin{align}
M_{pass} = \ & \frac{1}{3} \Big[ \frac{\mathds{1} \otimes \mathds{1} + \sigma_X \otimes \sigma_X}{2}  
 +   \frac{\mathds{1} \otimes \mathds{1}  + \sigma_Z \otimes \sigma_Z }{2}  \nonumber \\
& +  \frac{\mathds{1} \otimes \mathds{1} + \sigma_Z \sigma_X \otimes \sigma_Z \sigma_X}{2} \Big]  \nonumber \\
= \ & \frac{ \mathds{1} \otimes \mathds{1}  + 2\Pi }{3}.
\end{align}
Let us define $M_{ACC}$ as the POVM element that corresponds to accepting. In Protocol \ref{alg:authteleptrust}, we see that all tests on copies $i \neq r$ must pass in order to do so. This gives
\begin{align}
M_{ACC} = \underset{i \neq r}{\otimes} M_{pass_{i}}.
\end{align}

\begin{theorem}[{\cite{Marin}}]
In the case of perfect measurements, Protocol \ref{alg:authteleptrust} has completeness 1 and soundness $\frac{1}{S}$.
\end{theorem}

\begin{proof}
First, in the case that the source supplies all ideal Bell states, we show that the protocol works perfectly.
We write $S$ copies of the ideal state as 
$ \underset{S}{\otimes} \Pi$. 
We can then calculate the probability of successful authenticated teleportation using copy $r$, when the source supplies all ideal states, as 
\begin{align}
\text{Tr} \Big[  \Pi_r \otimes M_{ACC}  \underset{S}{\otimes} \Pi \Big] 
& = \text{Tr} \Big[ \Pi_r \underset{i \neq r}{\otimes} \Big( \frac{\mathds{1} \otimes \mathds{1} +  2\Pi}{3} \Big)_i \Pi_i  \Big] \nonumber \\
 & = \text{Tr} \Big[ \Pi_r \underset{i \neq r}{\otimes}  \Pi_i \Big] \nonumber \\
& = 1.
\end{align}
As we see, in the ideal case we will always pass the verification test, and using the perfect Bell state for teleportation results in a teleportation fidelity of 1. Thus, if the source supplies ideal states, the parties can perform authenticated teleportation of a quantum message perfectly.

We now consider the soundness bound. 
Let $\rho_B$ be the output of the protocol, which is the teleported qubit on Bob's side along with the classical register that indicates whether to accept or reject the output. 
It is known that the fidelity of the teleportation is at least as high as the fidelity of the entangled state with the ideal Bell state \cite{Horodecki}.

Let $P_{fail}$ be the projector onto the orthogonal subspace of the qubit $\ket{\psi}$ that Alice wants to teleport, given that the teleportation is accepted. If Bob's state $\rho_B$ at the end of the protocol belongs to this subspace, the protocol has failed. 
We have
\begin{align}
\text{Tr}(P_{fail} \rho_B) & =  \text{Tr} \Big[(\mathds{1} - \ket{\psi}\bra{\psi} ) \otimes \ket{ACC} \bra{ACC} \rho_B \Big] \nonumber \\
& \leq \text{Tr} \Big[  (\mathds{1} \otimes \mathds{1} - \Pi)_r \otimes \ket{ACC}\bra{ACC} \rho_{AB}^r \Big], 
\end{align}
where the total output state after the verification steps, $\rho_{AB}^r$, is given by
\begin{align}
\rho_{AB}^r = \ & \frac{1}{S} \sum_{r=1}^S \Big[ p_{acc} \rho_{acc}^r \otimes \ket{ACC}\bra{ACC} \nonumber \\
& +  p_{rej} \rho_{rej}^r \otimes \ket{REJ}\bra{REJ} \Big],
\end{align}
with $p_{acc}, p_{rej}$ denoting the probability of accepting or rejecting, and $\rho_{acc}^r, \rho_{rej}^r$ denoting the output states conditioned on accepting or rejecting, respectively. 
When we condition on accepting, this becomes
\begin{align}
\text{Tr} (P_{fail} \rho_B)  
& \leq \text{Tr} \Bigg[  \frac{1}{S} \sum_{r=1}^S (\mathds{1} \otimes \mathds{1} - \Pi)_r p_{acc} \rho_{acc}^r \Bigg].
\end{align}
We will now denote the total state shared over all $S$ copies between Alice and Bob as $\rho_{1...S}$.
Then, the entangled state used to teleport, which is the post-measurement state conditioned on accepting, is given by
\begin{align}
\rho_{acc}^r 
& = \frac{1}{\text{Tr} \Big[ M_{ACC} \rho_{1...S} \Big]} \ \underset{i \neq r} {\text{Tr}}\Big[ M_{ACC} \rho_{1...S}\Big] \nonumber \\
& = \frac{1}{p_{acc}} \ \underset{i \neq r} {\text{Tr}}\Big[ M_{ACC} \rho_{1...S}\Big]. 
\end{align}
This gives 
\begin{align}
\text{Tr} (P_{fail} \rho_B) 
& \leq \text{Tr}  \Bigg[\frac{1}{S} \sum_{r=1}^S (\mathds{1} \otimes \mathds{1} - \Pi)_r  \otimes M_{ACC} \rho_{1...S} \Bigg].
\end{align}
Denoting 
\begin{align}
Q 
& = \frac{1}{S} \overset{S}{\underset{r=1}{\sum}}(\mathds{1}\otimes \mathds{1} - \Pi)_r \otimes M_{ACC},
\label{eq:qandy}
\end{align}
we have 
\begin{align}
\text{Tr}(P_{fail} \rho_B) \leq \text{Tr} (Q \rho_{1...S}).
\end{align}
We can determine an upper bound on this expression, no matter what state $\rho_{1...S}$ the source supplies, by computing the maximum eigenvalue of $Q$. 
In Protocol \ref{alg:authteleptrust}, we accept if all tests on copies $i \neq r$ pass, and so in this case $Q$ is given by
\begin{align}
Q
& = \frac{1}{S} \overset{S}{\underset{r=1}{\sum}}(\mathds{1}\otimes \mathds{1} - \Pi)_r \underset{i \neq r}{\otimes} \Big( \frac{ \mathds{1} \otimes \mathds{1}  + 2\Pi}{3}\Big)_i.
\end{align}
We know $\Pi$ is an eigenprojector of $\frac{\mathds{1} \otimes \mathds{1} + 2\Pi}{3}$. Let $\Pi^\perp$ be the projector $(\mathds{1} \otimes \mathds{1} - \Pi)$. Then, the complete set of eigenprojectors for $Q$ is given by  
\begin{align}
\{ \underset{l \neq m}{\underset{k,}{\otimes}} \Pi^\perp_l \underset{m \neq l}{\underset{S-k,}{\otimes}} \Pi_m \},
\label{eq:qeig}
\end{align}
where $k \in \{0, 1, ...,  S\}$ is the number of $\Pi^\perp$'s in the eigenprojector.
We must then determine an expression for the eigenvalues of $Q$ as a function of $k, S$, which we will denote as $g(k, S)$, from the eigenvalue equation 
\begin{align}
Q \Big[ \underset{l \neq m}{\underset{k,}{\otimes}} \Pi^\perp_l \underset{m \neq l}{\underset{S-k,}{\otimes}} \Pi_m \Big] = g(k, S) \Big[ \underset{l \neq m}{\underset{k,}{\otimes}} \Pi^\perp_l \underset{m \neq l}{\underset{S-k,}{\otimes}} \Pi_m \Big].
\end{align}
We will make use of the following: 
\begin{align}
(\mathds{1} \otimes \mathds{1} - \Pi) \Pi  = 0,  & \text{ \  } (\mathds{1} \otimes \mathds{1} - \Pi) \Pi^\perp  = \Pi^\perp, \nonumber \\
\Big( \frac{ \mathds{1} \otimes \mathds{1} + 2\Pi}{3} \Big) \Pi  = \Pi, &  \text{ \ } \Big( \frac{ \mathds{1} \otimes \mathds{1} + 2\Pi}{3} \Big) \Pi^\perp  = \frac{1}{3} \Pi^\perp .
\end{align}
By examining the action of $Q$ on an eigenprojector with $(S-k)$ $\Pi$'s and $k$ $\Pi^\perp$'s, we see that the corresponding eigenvalue is given by 
\begin{equation}
g(k,S) = \frac{k}{S} (1)^{S-k} \Big( \frac{1}{3} \Big)^{k-1} 
= \frac{k}{S}  \frac{1}{3^{k-1}}.
\end{equation}
Here, $\frac{k}{S}$ is the probability of a randomly chosen eigenprojector $r$ belonging to the set of $k$ $\Pi^\perp$ terms (note that if $r$ belongs to the set of $\Pi$ terms, then this does not contribute to the eigenvalue), while the remaining terms come from the action of $M_{ACC}$ on $\Pi, \Pi^\perp$.
To determine the soundness bound, we then find the maximum value of $g(k, S)$ over all $k$, which occurs for $k=1$, giving
\begin{equation}
\text{Tr} (P_{fail} \rho_B)  \leq \max_k \  \frac{k}{3^{k-1} S} = \frac{1}{S}.
\label{eq:soundfirst}
\end{equation}

\end{proof}

This shows that the optimal cheating strategy for a dishonest source is to provide all copies but one as ideal Bell states. The probability of the non-ideal copy being used for teleportation is then $\frac{1}{S}$, representing the maximum probability of failure of our protocol. 

Now, given that the parties have accepted the transmission of the qubit as valid,
we can derive an expression for the fidelity of the teleported qubit, $f = \text{Tr} (\ket{\psi} \bra{\psi}  \rho_{acc}^B)$, in terms of this soundness bound and the probability of acceptance:
\begin{align}
\text{Tr} (P_{fail} \rho_B ) 
& = \text{Tr} \Big[ (\mathds{1} - \ket{\psi}\bra{\psi} ) \otimes \ket{ACC} \bra{ACC} \rho_B \Big] \nonumber \\
& = \text{Tr} \Big[ (\mathds{1} - \ket{\psi} \bra{\psi}) p_{acc} \rho_{acc}^B \Big] \nonumber \\
& = p_{acc} ( 1 - f ).
\label{eq:fidexpr}
\end{align}
Thus, using our calculated expression for soundness in Equation (\ref{eq:soundfirst}), we obtain the following expression for the fidelity of the teleported qubit in Protocol \ref{alg:authteleptrust}: 
\begin{align}
f  = 1 - \frac{\text{Tr} (P_{fail} \rho_B )}{p_{acc}} \geq 1 - \frac{1}{S p_{acc}}.
\end{align}

\subsection{Noisy measurements}

We now extend this to the case where we have imperfect, or noisy, measurements, and model the scenario by introducing a noise parameter $p \in [0, 1]$. 
Let us denote the POVM element for passing the test when we do the noisy $\sigma_X \otimes \sigma_X$ measurement as
\begin{equation}
p \frac{\mathds{1} \otimes \mathds{1} + \sigma_X \otimes \sigma_X}{2} + (1-p) \frac{\mathds{1} \otimes \mathds{1}}{2}, 
\end{equation}
and similarly for the other Pauli measurements $\sigma_Z \otimes \sigma_Z, \sigma_Z \sigma_X \otimes \sigma_Z \sigma_X$. Thus, with probability $p$ each measurement proceeds perfectly, and with probability $(1-p)$ the measurement randomly gives either a $\pm 1$ outcome. The POVM element for passing a test is then
\begin{align}
M_{pass} & = \frac{ (3-p) \mathds{1} \otimes \mathds{1} + 4 p \Pi}{6},
\end{align}
and the overall POVM element for accepting in Protocol \ref{alg:authteleptrust} is given by
\begin{align}
M_{ACC} = \underset{i \neq r}{\otimes} \Big( \frac{ (3-p) \mathds{1} \otimes \mathds{1} + 4 p \Pi}{6} \Big)_i .
\end{align}

\begin{theorem}
In the case of noisy measurements with noise parameter $p$, Protocol \ref{alg:authteleptrust} has completeness $( \frac{1+p}{2} )^{S-1}$ and soundness $ \underset{k}{\max \ }  \frac{k}{S} (\frac{1+p}{2})^{S-k} ( \frac{3-p}{6})^{k-1}$, where $k \in \{0, 1,...,  S\}, p \in [0,1]$.
\end{theorem}

\begin{proof}
We will use:
\begin{align}
(\mathds{1} \otimes \mathds{1} - \Pi) \Pi  = 0, \text{ \ } ( &\mathds{1} \otimes \mathds{1}  - \Pi) \Pi^\perp   = \Pi^\perp,  \nonumber \\
\Big( \frac{ (3-p) \mathds{1} \otimes \mathds{1} + 4 p \Pi}{6} \Big)\Pi & = \Big( \frac{1+p}{2} \Big) \Pi, \nonumber \\ 
\Big( \frac{ (3-p) \mathds{1} \otimes \mathds{1} + 4 p \Pi}{6} \Big) \Pi^\perp & = \Big( \frac{3-p}{6} \Big) \Pi^\perp . 
\label{eq:noisyp}
\end{align}
The probability of successful authenticated teleportation when the source supplies ideal states is given by
\begin{align}
\text{Tr} \Big[ \Pi_r & \otimes M_{ACC} \underset{S}{\otimes} \Pi \Big] \nonumber \\
& = \text{Tr} \Big[ \Pi_r \underset{i \neq r}{\otimes} \Big( \frac{ (3-p) \mathds{1} \otimes \mathds{1} + 4 p \Pi}{6} \Big)_i \Pi_i  \Big] \nonumber \\
& = \text{Tr} \Big[ \Pi_r  \underset{i \neq r}{\otimes} \Big( \frac{1+p}{2} \Big) \Pi_i  \Big] \nonumber \\
&  =  \Big( \frac{1+p}{2} \Big)^{S-1} .
\end{align}
To calculate soundness, we replace $Q$ in the previous proof by substituting in Equation (\ref{eq:qandy}) our new expression for $M_{ACC}$, giving
\begin{equation}
Q = \frac{1}{S} \sum_{r=1}^S (\mathds{1} \otimes \mathds{1} - \Pi)_r \underset{i \neq r}{\otimes} \Big( \frac{ (3-p) \mathds{1} \otimes \mathds{1} + 4 p \Pi}{6} \Big)_i.
\end{equation}
Considering the action of $Q$ on a general eigenprojector with $(S-k)$ $\Pi$'s and $k$ $\Pi^\perp$'s, 
we find the expression for the corresponding eigenvalue to be
\begin{equation}
 g(k,S,p) = \frac{k}{S} \Big(\frac{1+p}{2}\Big)^{S-k} \Big( \frac{3-p}{6}\Big)^{k-1} .
\end{equation}
Again, we must determine the maximum eigenvalue of $Q$ in order to bound the soundness of the protocol. This gives
\begin{align}
\text{Tr} (P_{fail} \rho_B) \leq \max_k \ \frac{k}{S} \Big(\frac{1+p}{2}\Big)^{S-k} \Big( \frac{3-p}{6}\Big)^{k-1}.
\end{align}

\end{proof}

We can compute this bound numerically. In Figure \ref{fig:noisysoundness}, we plot the maximum eigenvalue of $Q$ for different values of $p$, along with the completeness bound. This seems to show that the protocol can never fail in high noise. However, this is because our condition on the protocol accepting the state is that every copy must pass the test (Step 5). In high noise cases, we saw in our discussion of completeness that the probability of each state passing the test is low, and so it is very unlikely that noisy measurements will allow every test to pass. Thus, the protocol hardly ever accepts the state in high noise, and so it hardly ever fails. 

From our analysis, it is clear that the protocol does not work well in the setting of noisy measurements. Therefore, we consider a modification of the protocol in the next section.

\begin{figure}[t]
\centering
\includegraphics[trim = 1cm 0cm 0cm 1cm, width=0.5\textwidth]{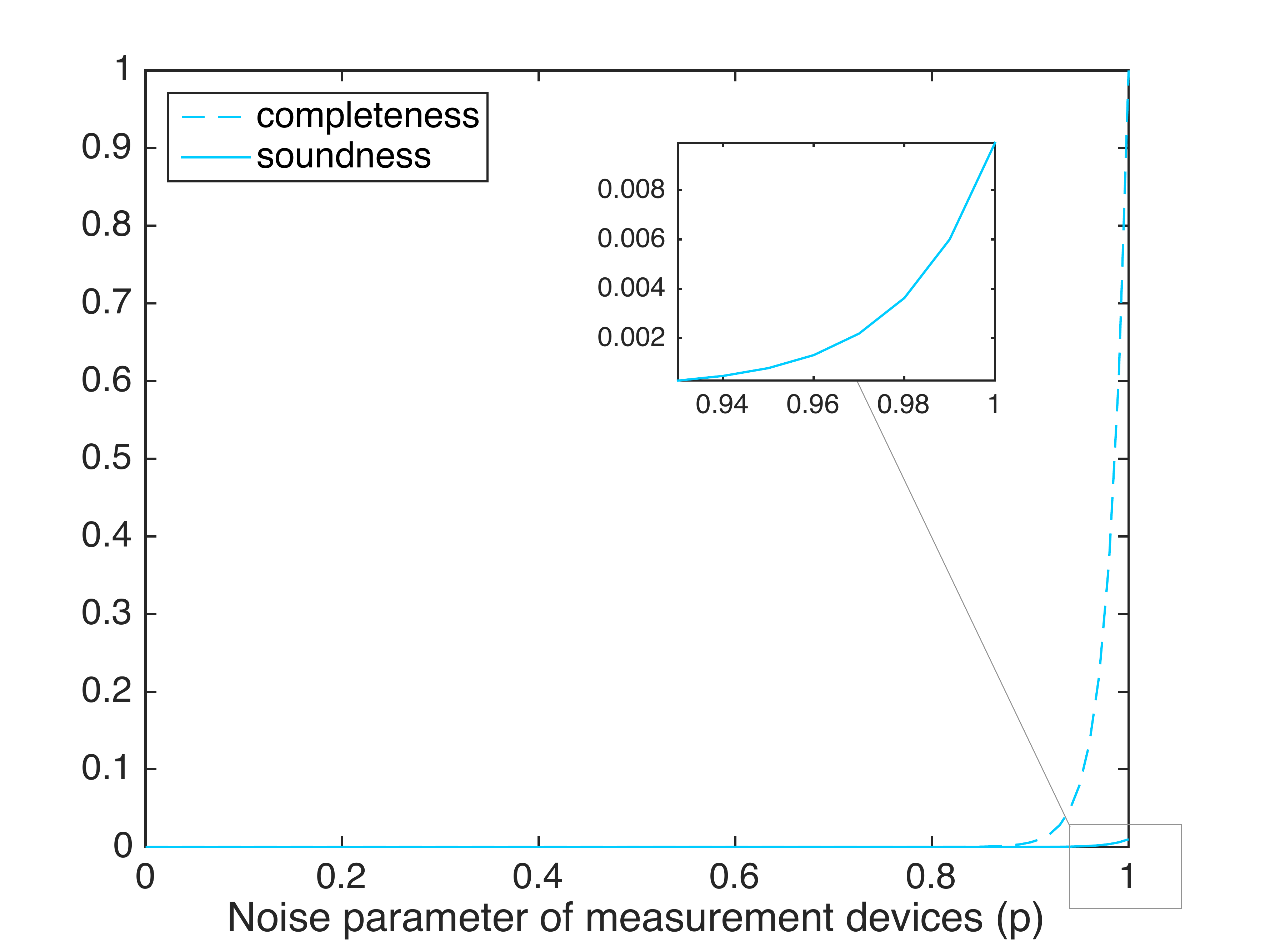}
\caption{Variation in completeness (dashed line) and soundness (solid line) with noise parameter of the measurements in Protocol \ref{alg:authteleptrust}. The inset shows the soundness bound for high values of $p$. $S=101$. \label{fig:noisysoundness}}
\end{figure}

\subsection{Noisy measurements with relaxed accept condition}
We now consider no longer requiring each and every copy that is tested to pass, in order to accept the quantum message teleported using copy $r$.
Let us denote the maximum number of tests that can fail as $\Delta \in \{0, 1, ... , S-1\}$; this means that now, at least $S-1-\Delta$ tests must pass in order to accept. A modified version of the protocol incorporating such a failure threshold is given in Protocol \ref{alg:authtelepthres}.

In our analysis, we must consider all possible combinations of tests that are allowed to fail. The POVM element for passing a test remains $M_{pass} = \big( \frac{ (3-p) \mathds{1} \otimes \mathds{1} + 4 p \Pi}{6} \big)$, while for failing a test it is $M_{fail} = \big( \frac{ (3+p) \mathds{1}\otimes \mathds{1} - 4 p \Pi}{6} \big)$.
We must then determine a new expression for $M_{ACC}$, which corresponds to accepting in Protocol \ref{alg:authtelepthres}.
Let us take $\mathcal{D}$ to be the set of tests that are allowed to fail, such that $|\mathcal{D}| \leq \Delta$, and consider all possible choices for the copies that belong to this set. This gives
\begin{align}
M_{ACC} & = \sum_{\substack{\mathcal{D}, \\ \abs{\mathcal{D}} \leq \Delta}} \underset{\substack{i \notin \mathcal{D}, \\ i \neq r}}{\otimes} M_{pass_{i}} \underset{\substack{i \in \mathcal{D}, \\ i \neq r}}{\otimes} M_{fail_{i}},
\label{eq:maccy}
\end{align}
which for our case of noisy measurements gives
\begin{align}
M_{ACC}
= \ & \sum_{\substack{\mathcal{D}, \\ \abs{\mathcal{D}} \leq \Delta}} \underset{\substack{i \notin \mathcal{D}, \\ i \neq r}}{\otimes} \Big( \frac{ (3-p) \mathds{1}\otimes \mathds{1} + 4 p \Pi}{6} \Big)_i \nonumber \\
& \underset{\substack{i \in \mathcal{D}, \\ i \neq r}}{\otimes} \Big( \frac{ (3+p) \mathds{1} \otimes \mathds{1} - 4 p \Pi}{6} \Big)_i.
\label{eq:noisymaccy}
\end{align}
 
 \begin{theorem}
In the case of noisy measurements with noise parameter $p$, Protocol  \ref{alg:authtelepthres} has completeness 
\begin{align}
\sum_{x=0}^{\Delta} {{S-1}\choose{x}} \Big( \frac{1+p}{2} \Big)^{S-1-x} \Big( \frac{1-p}{2} \Big)^x
\end{align}
and soundness
\begin{align}
\underset{k}{\max \ } \frac{k}{S} \ \sum_{x=0}^{\Delta} \ \sum_{y=0}^{\Delta - x} \ & {{S-k}\choose{x}} {{k-1}\choose{y}}  \Big( \frac{1+p}{2} \Big)^{S-k-x} \Big( \frac{1-p}{2} \Big)^x \nonumber \\
& \times \Big( \frac{3-p}{6} \Big)^{k-1-y} \Big( \frac{3+p}{6} \Big)^y.
\label{eq:soundcloud}
\end{align}
\label{th:noisymeasdelta}
\end{theorem}

\begin{algorithm}[t]
\caption{\textsc{Authenticated teleportation with failure threshold}}
\begin{flushleft}
\textit{Input}: Security parameter $S$. \\
\textit{Goal}: Alice teleports state $\ket{\psi}$ to Bob through an authenticated channel.
\end{flushleft}
\begin{algorithmic}[1]
\STATE An untrusted source generates $S$ copies of the Bell state, and sends the shares of each to Alice and Bob.  \\ \
\STATE Alice chooses a random $r \in \{1, 2, ..., S\}$ and a failure threshold $\Delta \in \{ 0, 1, ..., S-1 \}$, and sends $r, \Delta$ to Bob.
 \\ \
\STATE For all copies $i \neq r$, Alice randomly chooses to measure either $\sigma_X, \sigma_Y, \sigma_Z$ on her part of the state. She tells Bob which operator she measured, and her measurement outcome. \\ \
\STATE For all copies $i \neq r$, Bob measures the same operator as Alice. For each copy, if the product of their measurement outcomes is $+1$ (or $-1$ when measuring $\sigma_Y$), they pass the test, otherwise they fail. \\ \
\STATE If \textit{at least} $S-1-\Delta$ tests on copies $i \neq r$ were passed, the parties ACCEPT. Otherwise, they REJECT. \\ \
\STATE The parties use copy $r$ as the entangled state to teleport $\ket{\psi}$ to Bob. 
\end{algorithmic}
\label{alg:authtelepthres}
\end{algorithm}

\begin{proof}
We will need the following:
\begin{align}
(\mathds{1} \otimes \mathds{1} - \Pi) \Pi  = 0,  \text{ \ } ( & \mathds{1} \otimes \mathds{1} - \Pi) \Pi^\perp   = \Pi^\perp, \nonumber \\
\Big( \frac{ (3-p) \mathds{1} \otimes \mathds{1} + 4 p \Pi}{6} \Big)\Pi & = \Big( \frac{1+p}{2} \Big) \Pi, \nonumber \\
\Big( \frac{ (3-p) \mathds{1} \otimes \mathds{1} + 4 p \Pi}{6} \Big) \Pi^\perp & = \Big( \frac{3-p}{6} \Big) \Pi^\perp, \nonumber \\
\Big( \frac{ (3+p) \mathds{1} \otimes \mathds{1} - 4 p \Pi}{6} \Big)\Pi & = \Big( \frac{1-p}{2} \Big) \Pi, \nonumber \\
\Big( \frac{ (3+p) \mathds{1} \otimes \mathds{1} - 4 p \Pi}{6} \Big) \Pi^\perp & = \Big( \frac{3+p}{6} \Big) \Pi^\perp. 
\label{eq:qtpi}
\end{align}
The completeness bound is given by
\begin{align}
\text{Tr} \Big[  \Pi_r & \otimes M_{ACC}  \underset{S}{\otimes} \Pi \Big] \nonumber \\
= \ & \text{Tr} \Bigg[ \sum_{\substack{\mathcal{D}, \\ \abs{\mathcal{D}} \leq \Delta}} \underset{\substack{i \notin \mathcal{D}, \\ i \neq r}}{\otimes} \Big(  \frac{ (3-p) \mathds{1} \otimes \mathds{1} + 4 p \Pi}{6} \Big)_i \Pi_i \nonumber \\
& \underset{\substack{i \in \mathcal{D}, \\ i \neq r}}{\otimes} \Big(  \frac{ (3+p) \mathds{1} \otimes \mathds{1} - 4 p \Pi}{6} \Big)_i \Pi_i \Bigg] \nonumber \\
= \ & \sum_{x=0}^{\Delta} {{S-1}\choose{x}} \Big( \frac{1+p}{2} \Big)^{S-1-x} \Big( \frac{1-p}{2} \Big)^x.
\end{align}
To calculate the soundness bound, our expression for $Q$ from Equation (\ref{eq:qandy}) is given by 
\begin{align}
Q = \frac{1}{S}  \sum_{r=1}^S (\mathds{1} \otimes \mathds{1} - & \Pi)_r \
 \sum_{\substack{\mathcal{D}, \\ \abs{\mathcal{D}} \leq \Delta}} \
\underset{\substack{i \notin \mathcal{D}, \\ i \neq r}}{\otimes} \Big( \frac{ (3-p) \mathds{1} \otimes \mathds{1} + 4 p \Pi}{6} \Big)_i \nonumber \\
& \underset{\substack{i \in \mathcal{D}, \\ i \neq r}}{\otimes} \Big(  \frac{ (3+p) \mathds{1} \otimes \mathds{1} - 4 p \Pi}{6} \Big)_i  .
\end{align}
The eigenprojectors of $Q$ are again given by Equation (\ref{eq:qeig}).
From the action of $Q$ on an eigenprojector with $(S-k)$ $\Pi$ terms and $k$ $\Pi^\perp$ terms, using Equation (\ref{eq:qtpi}), we determine the corresponding eigenvalue expression to be
\begin{align}
g(k, S, p, \Delta) = \ & \frac{k}{S} \ \sum_{x = 0}^{\Delta} \ \sum_{y=0}^{\Delta - x} \ {{S-k}\choose{x}} {{k-1}\choose{y}} \nonumber \\
& \times \Big( \frac{1+p}{2} \Big)^{S-k-x} \Big( \frac{1-p}{2} \Big)^x \nonumber \\
& \times \Big( \frac{3-p}{6} \Big)^{k-1-y} \Big( \frac{3+p}{6} \Big)^y,
\end{align}
which leads to the soundness being the maximum of this expression over all possible values of $k$.
(Note that if the parties run Protocol \ref{alg:authtelepthres} with perfect measurement devices, the corresponding security bounds can be determined by substituting $p=1$ in the above expressions.)

\end{proof}

The maximum eigenvalue of $Q$ for a certain $\Delta$ can be determined numerically for various values of the noise parameter $p$. We plot this, along with the completeness bound, for different failure thresholds in Figure \ref{fig:noisymeas}. 
We see that allowing a large proportion of tests to fail of course gives a better completeness bound, but it comes at a cost of the protocol being more susceptible to cheating by a dishonest source. 
Our results demonstrate the robustness of our protocol, by providing a quantitative assessment of the tradeoff between how well the protocol works, and how likely it is to fail. This is particularly useful for an experimental implementation: for example, if the noise parameter of the measurement devices is known, we can use Theorem \ref{th:noisymeasdelta} to determine the appropriate failure threshold required to achieve our desired tradeoff.

\begin{figure}[t]
\centering
\includegraphics[trim = 10mm 7.5cm 0mm 8cm, width=0.5\textwidth]{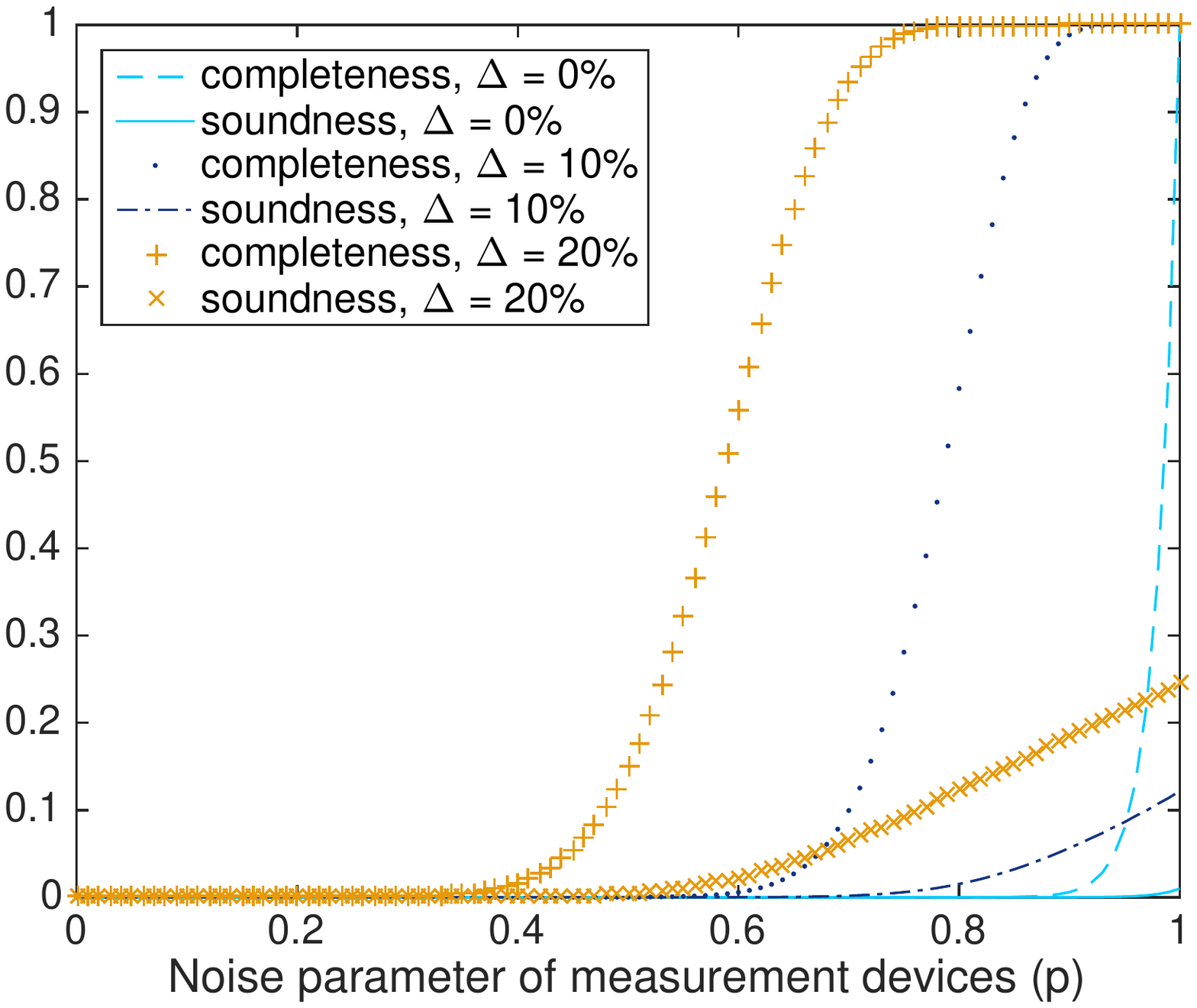}
\caption{Comparison of completeness and soundness bounds for different failure thresholds $\Delta$ in Protocol \ref{alg:authtelepthres}. $S=101$. 
\label{fig:noisymeas}}
\end{figure}

\section{Analysis for noisy states}

Let us now consider the scenario where, due to imperfections of realistic networks, the states prepared by the source are noisy versions of the Bell state. In such a setting, we are interested in analysing how noise affects the working of the protocol: how likely it is that the final state will be accepted, and the statements we can make about the fidelity of teleportation in this case.
We model the scenario using a generalised form of Werner states given by 
\begin{align}
\rho = v \Big[ (1-\eta) \ket{\Phi^+} \bra{\Phi^+} + \eta \ket{\Phi^-} \bra{\Phi^-} \Big] + (1-v) \frac{\mathds{1} \otimes \mathds{1}}{4},
\label{eq:noisay}
\end{align}
where $v$ is a parameter related to the amount of unpolarised noise in the experimental setup, and $\eta$ is related to the dephasing noise. These states, used for example in experiments such as \cite{Orieux2017}, take into account realistic noise that may be present when performing our protocols, and model polarisation-entangled photon pairs that can easily be created in the lab.
Note that when $\eta = 0$, $\rho$ takes the form of the Werner state \cite{Werner1989} with visibility $v$.

Let us first assume the measurements work perfectly for simplicity. As we saw earlier, in this case,  $M_{pass}  = \frac{ \mathds{1}\otimes \mathds{1}  + 2\Pi}{3}$. 
Then, the probability of passing a test with such a noisy state is 
\begin{align}
\text{Tr} ( M_{pass} \rho) 
=  \frac{3 + 3v - 4v\eta}{6}.
\end{align}
Let us first see how Protocol \ref{alg:authteleptrust} works in this scenario, where recall that $M_{ACC} = \underset{i \neq r}{\otimes} M_{pass_{i}}$.
The probability of accepting (passing all $S-1$ tests) when the source produces noisy states is 
\begin{align}
\text{Tr} (M_{ACC} \underset{S}{\otimes} \rho) 
& = \text{Tr} \Big[  \rho_r \underset{i \neq r}{\otimes} \Big( \frac{\mathds{1} \otimes \mathds{1} + 2 \Pi}{3} \Big)_i \rho_i \Big] \nonumber \\
& = \Big( \frac{3 + 3v - 4v\eta}{6} \Big)^{S-1}.
\end{align}

We now examine the probability of acceptance with noisy states in Protocol \ref{alg:authtelepthres}. In the case of perfect measurements, we have $M_{fail}  = \frac{ 2\mathds{1}\otimes \mathds{1}  - 2\Pi}{3}$. The modified $M_{ACC}$ is then
\begin{align}
M_{ACC} & = \sum_{\substack{\mathcal{D}, \\ \abs{\mathcal{D}} \leq \Delta}} \underset{\substack{i \notin \mathcal{D}, \\ i \neq r}}{\otimes} M_{pass_{i}} \underset{\substack{i \in \mathcal{D}, \\ i \neq r}}{\otimes} M_{fail_{i}} \nonumber \\
& = \sum_{\substack{\mathcal{D}, \\ \abs{\mathcal{D}} \leq \Delta}} \underset{\substack{i \notin \mathcal{D}, \\ i \neq r}}{\otimes} \Big( \frac{ \mathds{1} \otimes \mathds{1} + 2  \Pi}{3} \Big)_i \underset{\substack{i \in \mathcal{D}, \\ i \neq r}}{\otimes} \Big( \frac{ 2 \mathds{1} \otimes \mathds{1} - 2  \Pi}{3} \Big)_i.
\end{align}
We will also require the probability of failing a test with the noisy state:
\begin{align}
\text{Tr} (M_{fail} \rho) 
=  \frac{3 - 3v + 4v\eta}{6}.
\end{align}
%
Then, the acceptance probability in Protocol \ref{alg:authtelepthres} is given by 
\begin{align}
\text{Tr} (& M_{ACC} \underset{S}{\otimes} \rho) \nonumber \\
= \ & \text{Tr} \Bigg[ \rho_r \sum_{\substack{\mathcal{D}, \\ \abs{\mathcal{D}} \leq \Delta}} \underset{\substack{i \notin \mathcal{D}, \\ i \neq r}}{\otimes} \Big(  \frac{  \mathds{1} \otimes \mathds{1} + 2 \Pi}{3} \Big)_i \rho_i \nonumber \\
& \underset{\substack{i \in \mathcal{D}, \\ i \neq r}}{\otimes} \Big(  \frac{ 2 \mathds{1} \otimes \mathds{1} - 2 \Pi}{3} \Big)_i \rho_i \Bigg] \nonumber \\
= \ & \sum_{x=0}^{\Delta} {{S-1}\choose{x}} \Big( \frac{3 + 3v - 4v\eta}{6} \Big)^{S-1-x} \Big(\frac{3 - 3v + 4v\eta}{6} \Big)^x.
\end{align}

Note that for Werner states ($\eta = 0$), the above expression reduces to the completeness bound with noisy measurement devices. A plot of the probability of acceptance
for Werner states, motivated by experiments like \cite{Orieux2017}, 
is given in Figure \ref{fig:noisystatepic}, as a function of the failure threshold and visibility. If the source creates states of higher visibility, it is of course more likely to lead to the teleported qubit being accepted; further, we can see how increasing the failure threshold helps this.
\begin{figure}[t]
\centering
\includegraphics[trim = 10mm 7.5cm 0mm 8cm, width=0.5\textwidth]{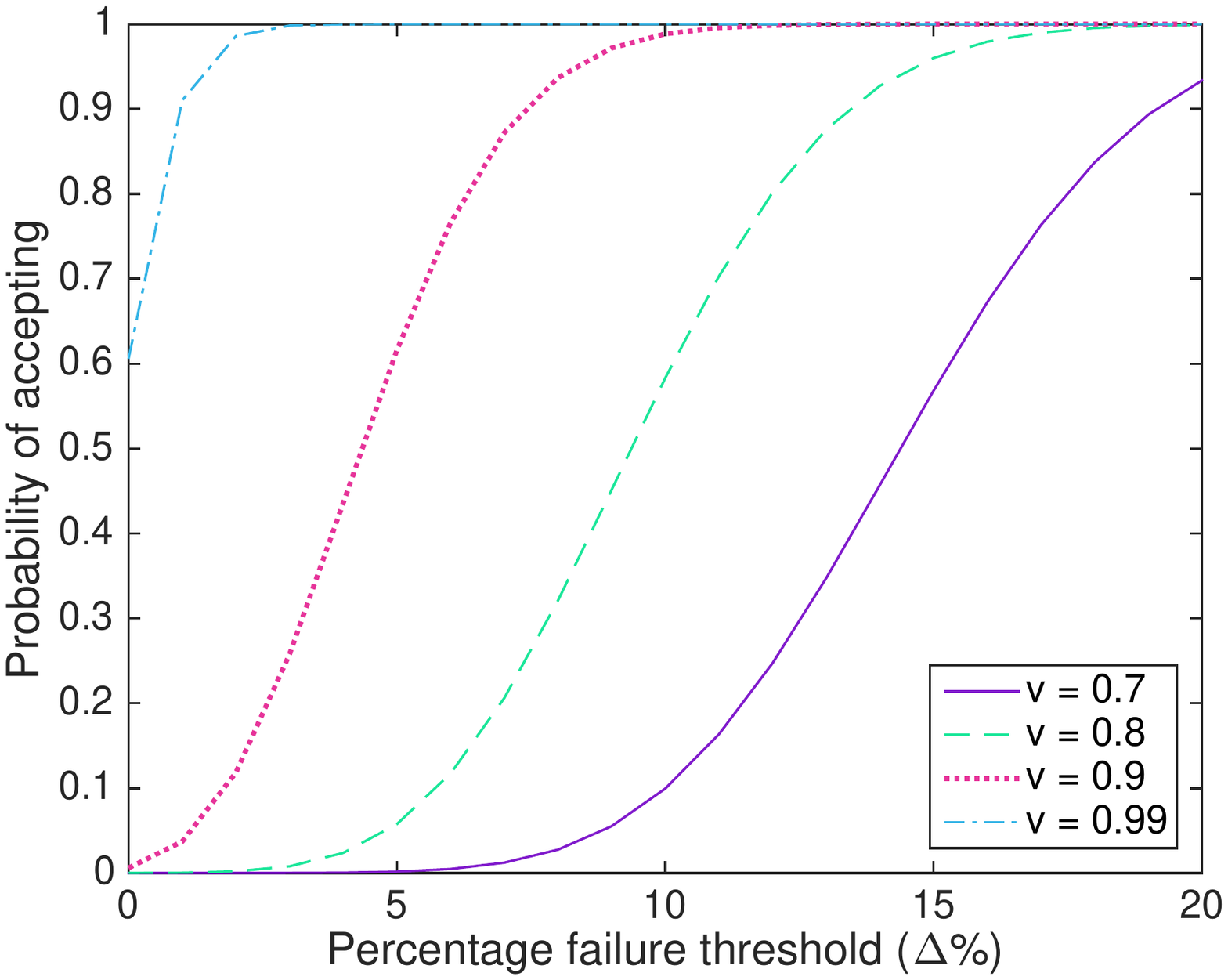}
\caption{Comparison of the variation of acceptance probability in Protocol \ref{alg:authtelepthres} with failure threshold $\Delta$, for Werner states of different visibilities $v$. $S = 101$.}
\label{fig:noisystatepic}
\end{figure}

Knowing the probability of acceptance when the source provides Werner states, we can determine a lower bound on the fidelity of teleportation in Step 6 of Protocol \ref{alg:authtelepthres}. For this, we require the soundness bound in the case of perfect measurements (Equation (\ref{eq:soundcloud}) with $p=1$), as well as Equation (\ref{eq:fidexpr}) for relating the soundness and probability of acceptance to the fidelity of teleportation. This gives

\small
\begin{align}
f \geq 1 - \frac{ \underset{k}{\max \ } \frac{k}{S} \ \overset{\Delta}{\underset{x=0}{\sum}} \ \overset{\Delta - x}{\underset{y=0}{\sum}} \ {{S-k}\choose{x}} {{k-1}\choose{y}} (1)^{S-k-x} (0)^x \big( \frac{1}{3} \big)^{k-1-y} \big( \frac{2}{3} \big)^y }{\overset{\Delta}{\underset{x=0}{\sum}} {{S-1}\choose{x}} \big( \frac{1+v}{2} \big)^{S-1-x} \big( \frac{1-v}{2} \big)^x}.
\end{align}
\normalsize

We plot this in Figure \ref{fig:grapho} for different failure thresholds and Werner state visibilities. For a source creating Werner states of visibility $v$ and allowing up to $\Delta$ failures, our results tell us the fidelity we can certify of an accepted transmission of a quantum message using Protocol \ref{alg:authtelepthres}. Further, by increasing the number of copies, the probability of acceptance increases while the soundness bound decreases, leading to a higher certified fidelity, and approaching the actual fidelity of teleportation with the Werner state.
\begin{figure}[t]
\centering
\includegraphics[trim = 10mm 7.5cm 0mm 8cm, width=0.5\textwidth]{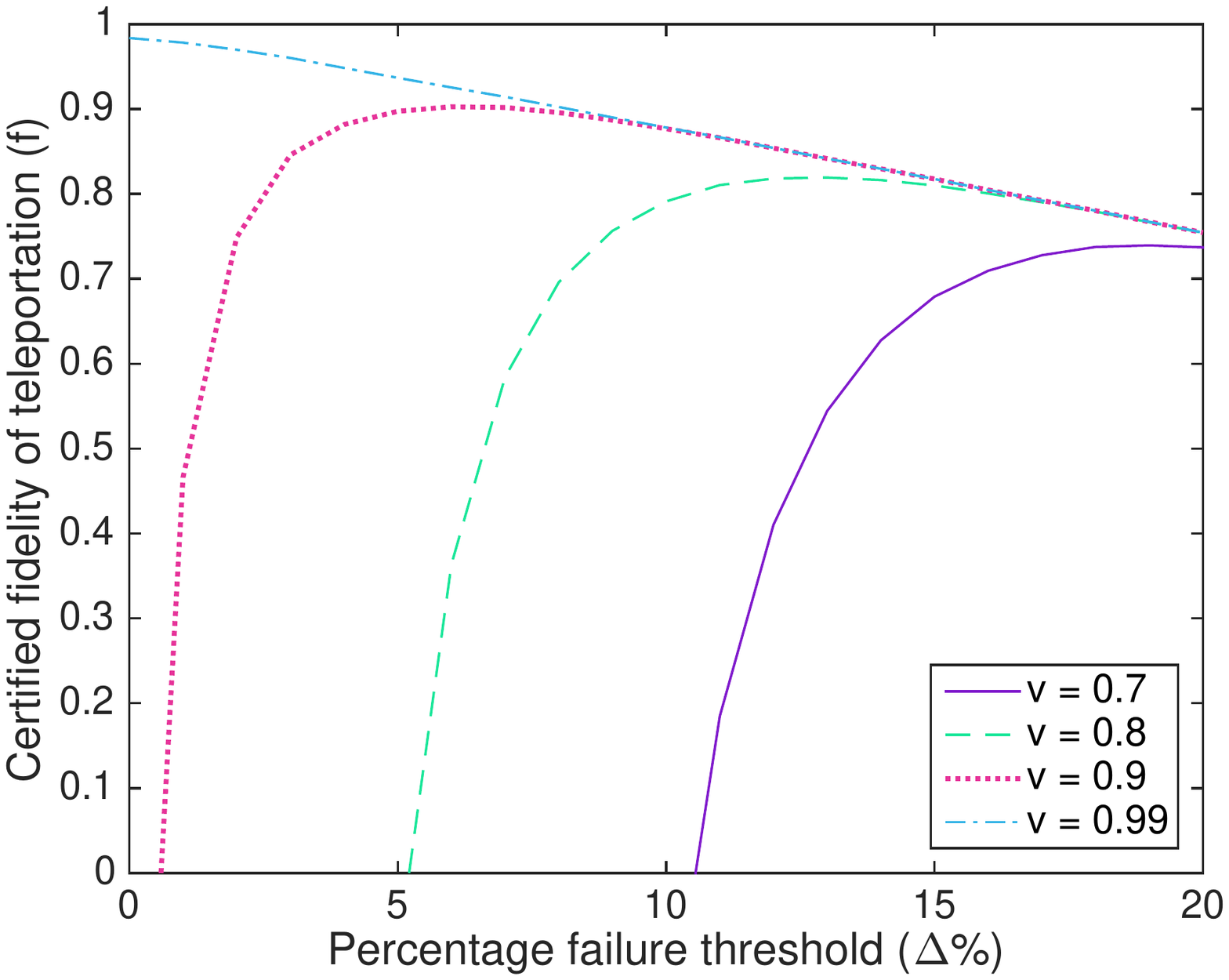}
\caption{Comparison of the variation of fidelity of teleportation certified using Protocol \ref{alg:authtelepthres} with failure threshold $\Delta$, for Werner states of different visibilities $v$. $S = 101$.}
\label{fig:grapho}
\end{figure}

It is then straightforward to extend this analysis to an experimental implementation with both noisy states and noisy measurements. Using Equation (\ref{eq:noisymaccy}) for the $M_{ACC}$ corresponding to noisy measurements and the generalised Werner state in Equation (\ref{eq:noisay}), we obtain
\begin{align}
\text{Tr} & (M_{ACC} \underset{S}{\otimes} \rho) \nonumber \\
& = \sum_{x=0}^{\Delta} {{S-1}\choose{x}} \Big( \frac{3 + 3pv - 4pv\eta}{6} \Big)^{S-1-x} \Big( \frac{3 - 3pv + 4pv\eta}{6} \Big)^x.
\end{align}
Setting $\eta = 0$ and using Equation (\ref{eq:soundcloud}), we get the following bound on the fidelity of teleportation with noisy measurements using Werner states:

\onecolumngrid
\begin{align}
f \geq 1 - \frac{\underset{k}{\max \ } \frac{k}{S} \ \overset{\Delta}{\underset{x=0}{\sum}} \ \overset{\Delta - x}{\underset{y=0}{\sum}} {{S-k}\choose{x}} {{k-1}\choose{y}}  \big( \frac{1+p}{2} \big)^{S-k-x} \big( \frac{1-p}{2} \big)^x \big( \frac{3-p}{6} \big)^{k-1-y} \big( \frac{3+p}{6} \big)^y}
{\overset{\Delta}{\underset{x=0}{\sum}} {{S-1}\choose{x}} \big( \frac{1+pv}{2} \big)^{S-1-x} \big( \frac{1-pv}{2} \big)^x}.
\end{align}
\\
\twocolumngrid
 

\section{Extension to graph state verification}

As demonstrated by Markham and Krause in \cite{Markham2018}, the technique of stabiliser-based verification can be naturally extended to graph states. We will now study the effect of noise on such a scenario. The verification of graph states is an important problem for quantum information, and has been studied in many works \cite{Hayashi2015, Takeuchi2018, Markham2018, Takeuchi2019, zhu2019efficient, zhu2019efficientb, zhu2019general, Unnikrishnan2020}. However, there are differences in the figures of merit used, which can make comparison difficult, and few results on analysing the effect of noise.

Recently, several works of Zhu and Hayashi have set up a general framework for verification of states, applied also to graph states (and hyper graph states) \cite{zhu2019efficient, zhu2019efficientb, zhu2019general}.
In these works, they employ different figures of merit than we do here; theirs arrives from the perspective of hypothesis testing, whereas our soundness statement is motivated by authentication.
Although these are related, and indeed one can bound the other (in both directions), an optimal bound in one is not optimal in another, and nor is the treatment of noise. 
Our bounds for the authentication figure of merit are tight and therefore optimal. 
The noiseless version of our protocol is also optimal in terms of the figures of merit in \cite{zhu2019efficient, zhu2019efficientb, zhu2019general}, as it is very similar to the protocol in \cite{Pallister18} which is proved to have optimal scaling in \cite{zhu2019general}, and indeed the same proof follows for \cite{Markham2018}.

In the following analysis, we will stick to studying the authentication-motivated figures of merit. This will have implications for the effect of noise on the figures of merit in \cite{zhu2019efficient, zhu2019efficientb, zhu2019general}, even if it does not necessarily bound them optimally. It would be interesting future work to investigate more thoroughly the effect of noise on other figures of merit.

In Protocol \ref{alg:graphvertrust}, we give a modified version of the Markham-Krause protocol \cite{Markham2018} that incorporates a failure threshold. Let us now see how a similar noise analysis affects the security bounds in this case. 

Consider an $n$-qubit graph state $\ket{G}$ shared between $n$ parties. Its full stabiliser group is given by $\mathcal{S} = \{ \mathcal{S}_j \}$, where $j = \{1, ..., 2^n \}$. 
Using the stabilisers for the graph state $\ket{G}$, we can write the projector onto the graph state as a combination of these,
\begin{equation}
 \Pi = \ket{G}\bra{G} =  \frac{1}{2^n} \sum_{j=1}^{2^n} \mathcal{S}_j .
\end{equation}
 The projector onto the $+1$ eigenstate of the stabiliser $\mathcal{S}_j$ (passing the test) is given by
$
\frac{\mathds{1} + \mathcal{S}_j}{2}
$. (Note that the identity matrix here is of size $2^n$.)
Let us model the noisy $\mathcal{S}_j$ measurement, using POVMs as before, as
\begin{align}
p \frac{\mathds{1} + \mathcal{S}_j}{2} + (1-p) \frac{\mathds{1}}{2}.
\end{align}
Then, the POVM element for passing a test is
\begin{align}
M_{pass} & = \frac{1}{2^n} \sum_{j=1}^{2^n} \Big[ p \Big( \frac{\mathds{1} + \mathcal{S}_j}{2} \Big) + (1-p) \frac{\mathds{1}}{2} \Big] \nonumber \\
& = \frac{1}{2^n} \sum_{j=1}^{2^n} p \frac{\mathds{1}}{2} + \frac{1}{2^n} \sum_{j=1}^{2^n} p \frac{\mathcal{S}_j}{2} + \frac{1}{2^n} \sum_{j=1}^{2^n} (1-p) \frac{\mathds{1}}{2} \nonumber \\
& = \frac{1}{2^n} p \frac{2^n \mathds{1}}{2} + \frac{p}{2} \frac{1}{2^n} \sum_{j=1}^{2^n} \mathcal{S}_j + \frac{1}{2^n} (1-p) \frac{2^n \mathds{1}}{2} \nonumber \\
& = \frac{ \mathds{1} + p\Pi}{2}, 
\end{align}
and the POVM element for failing a test is $
M_{fail} = \frac{ \mathds{1} - p \Pi}{2}$. 
We now write $M_{ACC}$, the overall POVM element for accepting in Protocol \ref{alg:graphvertrust}, as 
\begin{align}
M_{ACC} & =  \sum_{\substack{\mathcal{D}, \\ \abs{\mathcal{D}} \leq \Delta}} \ 
\underset{\substack{i \notin \mathcal{D}, \\ i \neq r}}{\otimes} M_{pass_{i}} \underset{\substack{i \in \mathcal{D}, \\ i \neq r}}{\otimes} M_{fail_{i}} \nonumber \\
& = \sum_{\substack{\mathcal{D}, \\ \abs{\mathcal{D}} \leq \Delta}} \ 
\underset{\substack{i \notin \mathcal{D}, \\ i \neq r}}{\otimes}\Big( \frac{ \mathds{1} + p \Pi}{2} \Big)_i \underset{\substack{i \in \mathcal{D}, \\ i \neq r}}{\otimes} \Big( \frac{\mathds{1} - p \Pi}{2} \Big)_i.
\label{eq:macc}
\end{align}
The analysis then proceeds similarly to before, with this new $M_{ACC}$. 
We give the most general result in the following Theorem, which can then be reduced to specific cases (eg. perfect measurements, no failure threshold) by substituting the relevant values of parameters. 

\begin{algorithm}[t]
\caption{\textsc{Verification of graph states with failure threshold} (adapted from \cite{Markham2018})}
\begin{flushleft}
\textit{Input}: Security parameter $S$. \\
\textit{Goal}: The parties verify that they share the $n$-qubit graph state $\ket{G}$.
\end{flushleft}
\begin{algorithmic}[1]
\STATE An untrusted source generates $S$ copies of the graph state, and sends the shares of each to the parties. \\ \
\STATE Party 1 chooses a random $r \in \{1, 2, ..., S\}$ and a failure threshold $\Delta \in \{ 0, 1, ..., S-1 \}$, and sends $r, \Delta$ to the other parties. \\ \
\STATE For all copies $i \neq r$, party 1 randomly chooses a stabiliser $\mathcal{S}_j$ to measure. She tells the other parties which stabiliser she chose, and her measurement outcome. \\ \
\STATE For all copies $i \neq r$, the other parties perform their corresponding stabiliser measurements. For each copy, if the product of all of the parties' measurement outcomes is $+1$, they pass the test, otherwise they fail. \\ \
\STATE If \textit{at least} $S-1-\Delta$ tests on copies $i \neq r$ were passed, the parties ACCEPT. Otherwise, they REJECT. \\ \
\STATE The parties use copy $r$ for their desired application.
\end{algorithmic}
\label{alg:graphvertrust}
\end{algorithm}

\begin{theorem}
In the case of noisy measurements with noise parameter $p$, Protocol \ref{alg:graphvertrust} has completeness 
\begin{align}
\sum_{x=0}^{\Delta} {{S-1}\choose{x}} \Big( \frac{1+p}{2} \Big)^{S-1-x} \Big( \frac{1-p}{2} \Big)^x
\end{align}
and soundness
\begin{align}
\underset{k}{\max \ } \frac{k}{2^{k-1} S} \ \sum_{x=0}^{\Delta} \ \sum_{y=0}^{\Delta - x} \ & {{S-k}\choose{x}} {{k-1}\choose{y}} \nonumber \\
& \times \Big( \frac{1+p}{2} \Big)^{S-k-x} \Big( \frac{1-p}{2} \Big)^x.
\end{align}
\end{theorem}

\begin{proof}
We now have
\begin{align}
(\mathds{1} - \Pi) \Pi  = 0, & \text{ \ } (\mathds{1} - \Pi) \Pi^\perp  = \Pi^\perp, \nonumber \\
\text{ \ } \Big( \frac{ \mathds{1}  + p \Pi}{2} \Big) \Pi  = \Big( \frac{1+p}{2} \Big) \Pi, & \text{ \ } 
\Big( \frac{ \mathds{1}  + p \Pi}{2} \Big) \Pi^\perp  = \frac{1}{2} \Pi^\perp, \nonumber \\
\text{ \ } \Big( \frac{ \mathds{1}  - p \Pi}{2} \Big) \Pi  = \Big( \frac{1-p}{2} \Big) \Pi, & \text{ \ } 
\Big( \frac{ \mathds{1}  - p \Pi}{2} \Big) \Pi^\perp  = \frac{1}{2} \Pi^\perp.
\end{align}
The completeness bound is given by
\begin{align}
\text{Tr} \Bigg[ \Pi_r &
\sum_{\substack{\mathcal{D}, \\ \abs{\mathcal{D}} \leq \Delta}} \
\underset{\substack{i \notin \mathcal{D}, \\ i \neq r}}{\otimes} \Big( \frac{ \mathds{1} + p\Pi}{2} \Big)_i \underset{\substack{i \in \mathcal{D}, \\ i \neq r}}{\otimes} \Big( \frac{ \mathds{1} - p \Pi}{2} \Big)_i  \underset{S}{\otimes} \Pi \Bigg] \nonumber \\
& = \sum_{x=0}^{\Delta} {{S-1}\choose{x}} \Big( \frac{1+p}{2} \Big)^{S-1-x} \Big( \frac{1-p}{2} \Big)^x.
\end{align}
To calculate the soundness bound, from Equation (\ref{eq:qandy}) our $Q$ is now given by
\begin{equation}
Q = \frac{1}{S}  \sum_{r=1}^S (\mathds{1}  - \Pi)_r  
\sum_{\substack{\mathcal{D}, \\ \abs{\mathcal{D}} \leq \Delta}} \
\underset{\substack{i \notin \mathcal{D}, \\ i \neq r}}{\otimes} \Big( \frac{\mathds{1} + p \Pi}{2} \Big)_i \underset{\substack{i \in \mathcal{D}, \\ i \neq r}}{\otimes} \Big( \frac{\mathds{1} - p \Pi}{2} \Big)_i  .
\end{equation}
In a similar way to the previous calculations, we determine a general expression for the eigenvalues of $Q$ as
\begin{align}
g(k, S, p, \Delta) = \ & \frac{k}{S} \ \sum_{x=0}^{\Delta} \ \sum_{y=0}^{\Delta - x} \ {{S-k}\choose{x}} {{k-1}\choose{y}} \nonumber \\
& \times \Big( \frac{1+p}{2} \Big)^{S-k-x} \Big( \frac{1-p}{2} \Big)^x \Big( \frac{1}{2} \Big)^{k-1-y} \Big( \frac{1}{2} \Big)^y \nonumber \\
= \ & \frac{k}{2^{k-1} S} \ \sum_{x=0}^{\Delta} \ \sum_{y=0}^{\Delta - x} \ {{S-k}\choose{x}} {{k-1}\choose{y}} \nonumber \\
& \times \Big( \frac{1+p}{2} \Big)^{S-k-x} \Big( \frac{1-p}{2} \Big)^x.
\end{align}
The soundness bound is then given by the maximum of this expression over all $k$.

\end{proof}

\section{Discussion}

In this work, we have analysed a simple approach to authenticated teleportation of a quantum message in a noisy network. We conclude by discussing the merits and drawbacks of our work, and future perspectives.

As in \cite{Marin, Barnum2002}, the interaction required in our protocol can be reduced by the parties sharing private random classical keys that encode the stabiliser they measure in each round, the copy $r$ that is used for the teleportation, and the allowed failure threshold $\Delta$. 

Comparing with previous authentication schemes for quantum messages \cite{Barnum2002, Hayden2016, Portmann2017}, our protocol, which consists of performing stabiliser measurements on many separate copies of the state, allows for a much easier experimental implementation. The existing schemes, on the other hand, use stabiliser error-correcting codes, with the number of qubits of the encoded state increasing with the desired security level. However, the tradeoff comes in terms of security, where both types of schemes are $\epsilon$-secure with $\epsilon = \frac{1}{S}$ for our protocol and scaling as $2^{-S}$ for protocols based on \cite{Barnum2002}. 

We will briefly comment on the difference between entanglement purification, used in \cite{Barnum2002}, and our method of authenticating the quantum channel in teleportation. In entanglement purification, one starts out with many copies of a state and applies operations to transform it to a (fewer) number of copies of a maximally entangled state. In the cryptographic setting we consider, an untrusted source is creating copies that we do not assume anything about (notably, they may be different in each round). This facilitates the application of our scheme to other scenarios where an adversarial source creates untrusted entanglement.

To demonstrate our protocol experimentally, one simply needs to generate many copies of a state, which is perfectly feasible with existing experimental setups. The noise tolerance of our protocol has been illustrated in terms of the tradeoff between the security bounds. 
In addition to considering noise in the authentication procedure, future work could also incorporate noise in the teleportation procedure itself.

Finally, we showed that our analysis for noisy measurement devices can be extended to verification of graph states shared between a network of parties. This has potential applications in the multitude of ways graph states are used across quantum information, be it for secret sharing \cite{Markham}, metrology \cite{Toth2014}, or blind quantum computation \cite{Gheorghiu2017}. 

\textit{Acknowledgements.} We thank Nathan Walk for helpful discussions. We acknowledge funding from the EPSRC, the ANR through the ANR-17-CE24-0035 VanQuTe project, and the European Union's Horizon 2020 Research and Innovation Programme under Grant Agreement No. 820445 (QIA).


\end{document}